\documentclass[conference]{IEEEtran} 
\usepackage{graphicx}
\usepackage{amsmath}
 \usepackage{amsthm} 
\usepackage{amsfonts}
\usepackage{amssymb}
\usepackage{mathtools}
\usepackage{dsfont}
\usepackage[toc,acronym,shortcuts]{glossaries}
\usepackage{graphicx}%
\usepackage{color}%
\usepackage{algorithmic}
 \usepackage{algorithm}
\newtheorem{mytheorem}{Theorem}

\newtheorem{lemma}{Lemma}
\usepackage[english]{babel}
\usepackage{blindtext}
\usepackage{epsfig}
\usepackage{framed}
\usepackage{xcolor}
\usepackage{color,soul}

\newacronym{6LoWPAN}{6LoWPAN}{IPv6 over Low-power Wireless Personal Area Networks}
\newacronym{CCM}{CCM}{Counter with CBC-MAC}
\newacronym{PDoS}{PDoS}{path-based DoS}
\newacronym{WPS}{WPS}{Wi-Fi Protected Setup}
\newacronym{CCA}{CCA}{Clear Channel Assessment}
\newacronym{OTG}{OTG}{On-The-Go}
\newacronym{PHY}{PHY}{Physical-Layer}
\newacronym{RSSI}{RSSI}{Received-Signal-Strength-Identification}
\newacronym{CBC-MAC}{CBC-MAC}{CBC-MAC}
\newacronym{MAC}{MAC}{MAC}
\newacronym{AES}{AES}{AES}

\IEEEoverridecommandlockouts

\title{On-The-Fly Secure Key Generation with Deterministic Models}

\author{
\IEEEauthorblockN{Rick Fritschek and Gerhard Wunder}
\IEEEauthorblockA{Heisenberg Communications and Information Theory Group\\
    Freie Universit\"at Berlin, \\
    Takustr. 9,
    D--14195 Berlin, Germany\\
    Email: rick.fritschek@fu-berlin.de, g.wunder@fu-berlin.de
    \thanks{This work was carried out within DFG grant WU 598/8-1 (DFG Priority Program on Compressed Sensing).}}
}
\begin{document}

\maketitle
\begin{abstract}
It is well-known that wireless channel reciprocity together with fading can be exploited to generate a common secret key between two legitimate communication partners. This can be achieved by exchanging known deterministic pilot signals between both partners from which the random fading gains can be estimated and processed. However, the entropy and thus quality of the generated key depends on the channel coherence time. This can result in poor key generation rates in a low mobility environment, where the fading gains are nearly constant. Therefore, wide-spread deployment of wireless channel-based secret key generation is limited. To overcome these issues, we follow up on a recent idea which uses unknown random pilots and enables ``on-the-fly'' key generation. In addition, the scheme is able to incorporate local sources of randomness but performance bounds are hard to obtain with standard methods. In this paper, we analyse such a scheme analytically and derive achievable key rates in the Alice-Bob-Eve setting. For this purpose, we develop a novel approximation model which is inspired by the linear deterministic and the lower triangular deterministic model. Using this model, we can derive key rates for specific scenarios. We claim that our novel approach provides an intuitive and clear framework to analyse similar key generation problems.

\end{abstract}


\section{Introduction}

New application scenarios such as the Internet of Things (IoT) and the Tactile Internet \cite{Wunder15} have sparked
interest in new security paradigms which are related to the physical layer. A typical security challenge e.g. in IoT is the design of
efficient and usable key management schemes for devices with high resource constraints \cite{Mukherjee15, Krentz15}. A famous problem in this context is the key-agreement between two parties, Alice and Bob, while being intercepted by an eavesdropper called Eve. The objective is to generate a key with a rate as large as possible while preventing Eve to get information about the key. The key must be the same for both Alice and Bob to be useful and thus forming a key-agreement problem. Systematic study of key-agreement and common randomness started with the work of \cite{Ahlswede93} and \cite{Maurer93}, where upper and lower bounds were found as well as insight which remained important. In \cite{Ahlswede93}, the key-agreement problem was split into two models, the source-type model and the channel-type model. In the former, each of the two terminals, Alice and Bob, has access to one component of a two component random source. In the later, both terminals exchange information via a wireless channel. In both cases, the terminals try to decode and compute a common key from their observations. Recent advances exploit the wireless channel to generate a common randomness at both terminals. In \cite{Ye10,Lai12} the concept of reciprocity was used to gain an advantage over Eve. Reciprocity refers to the phenomenon that the channel gain and therefore the channel, between two terminals, is nearly the same in both directions. If one considers a fading channel, then the channel gain represents another source of randomness. This can be estimated at both terminals with pilot signalling. Alice and Bob can agree on the same randomness, as long as perfect reciprocity and perfect estimation is assumed. These schemes even work under imperfect conditions, as long as both observations are correlated, by using Slepian-Wolf coding schemes over a public channel. However, the fading gain needs to provide sufficient randomness, which is dependent on the coherence time. A slow varying channel therefore results in a lower key-rate, which can even go to zero in the limit. Recent works on the subject try to overcome this problem by using relay-assisted schemes or mixed schemes, where the model is considered as source and channel-type dependent on the situation \cite{LaiPoorKeyGen}, \cite{Lai12}. Another approach was used in \cite{Wunder2016}, which utilized local randomness sources with a product signalling scheme to overcome the problem. However, the resulting rate expressions are non-trivial and closed-form solutions could not be obtained.

{\bf Contributions:} We develop a model, closely related to the linear deterministic model \cite{Avestimehr2011} and the lower triangular deterministic model \cite{Niesen-Ali}, to approximate the Gaussian channel model of \cite{Wunder2016}. The new approximate model has several advantages due to its properties, e.g. public communication being obsolete and build-in quantization. Using this model, we derive key rate results for several special cases of \cite{Wunder2016}. Moreover, one can compare the new method of product signalling with the classical pilot signalling. Specifically, we show that product signalling offers a more robust key generation in comparison to pilot signalling, since local randomness can be utilized to compensate the lack of randomness in the channel gain. Finally, we believe that the approximation model can be useful in other scenarios too, e.g multi-user key generation.

\section{Gaussian System Model} 

Alice communicates with Bob in a two-way non-duplex fashion. Looking at $n$ total channel uses, both alternate in receive and transmit mode such that Alice sends signals to Bob at the odd channel uses while, Bob utilizes the even channel uses. We denote the channel gain from Alice to Bob with $K$ and from Bob to Alice with $K'$. We assume reciprocity, meaning that within one communication round ($i$ and $i+1$), both $K$ and $K'$ are the same. Moreover, we assume that the channel gain is a fading parameter, changing randomly with a Gaussian distribution $K\sim\mathcal{N}(0,\sigma_K^2)$ after a full communication round (channel use $i$ and $i+1$). Alice and Bob have an additional local source of randomness, $\omega_A$ and $\omega_B$, respectively, which can be used for the inputs. Both communication channels are in presence of a wire-tapper Eve, which can receive Alice's input through a channel $H_1$ and Bob's input through a channel $H_2$. We can therefore write the channel equations in the following way 
 \begin{IEEEeqnarray*}{rCl}
Y_B&=& KX_A+Z_1\\
Y_A&=& KX_B+Z_2
\end{IEEEeqnarray*} and 
\begin{equation}
Y_E=\begin{cases}
H_{1}X_A & \text{for } n \text{ odd} \\
H_{2}X_B & \text{for } n \text{ even}.
\end{cases}
\end{equation}

The system model is illustrated in Fig.~\ref{System model2}. The figure also depicts the public noiseless channels $\Phi$ and $\Psi$, which are available for both Alice and Bob.
We have $t$ time instances in which we use the wireless channels $n$ times and the public noiseless channels $k=t-n$ times, where $n\leq t$. 

\begin{figure}
\centering
\includegraphics[scale=0.9]{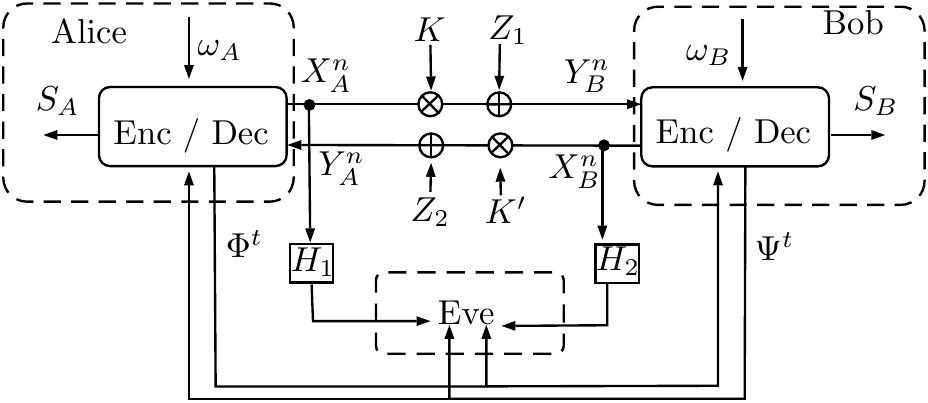}
\caption{Illustration of the system model with (dashed communication)  and without side information at Eve}
\label{System model2}
\end{figure}

There are $n$ rounds of wireless communication, where in each round $i$ Alice (Bob) sends a codeword $X_1(\omega_A,i)$ ($X_2(\omega_B,i)$) over the channel. We denote $X_1^n=(X_1(1),\cdots,X_1(n))$ and $X_2^n=(X_2(1),\cdots,X_2(n))$.

 Moreover, let $f_A$ and $f_B$ denote the key generation functions at Alice and Bob, respectively. We therefore have that the keys for Alice and Bob, are $S_A=f_A(X_A^n,Y_A^n,\Phi^k)$ and $S_B=f_B(X_B^n,Y_B^n,\Psi^k)$, respectively.

We define an achievable key rate $R_{key}$ if for every $\epsilon>0$ and sufficiently large $n$ there exists a strategy such that $S_A$ and $S_B$ satisfy
\begin{equation}
\Pr\{S_A\neq S_B\}<\epsilon,
\end{equation}
\begin{equation}\label{sndEQ}
\tfrac{1}{n} I(\Phi^k,\Psi^k,Y_E^n;S_A)<\epsilon,
\end{equation}
\begin{equation}
\tfrac{1}{n} H(S_A) > R_{key}-\epsilon,
\end{equation}
\begin{equation}
\tfrac{1}{n} \log |S_A| <\tfrac{1}{n} H(S_A)+\epsilon,\label{Unif}
\end{equation}
where $|S_A|$ denotes the alphabet size of the discrete key random variable $S_A$, see also \cite{Ahlswede93}.

It was shown in \cite{Ahlswede93} that if both terminals observe correlated source outputs $X^n$ and $Y^n$ from a discrete memoryless multiple source with generic sources $(X,Y)$, a secrecy key rate of $I(X;Y)$ can be achieved. The proof uses only a single forward or backward transmission of the public channel along with an extended Slepian-Wolf coding scheme. While the result was originally proved for discrete sources, it can be extended to continuous sources as well \cite{Ye06,Nitinawarat08}. Moreover, the result can be extended to the case of a pair of sources, for example $(X_A,Y_A,X_B,Y_B)$. To see this, one can use the same idea as in \cite{Ahlswede93}, in conjunction with the Slepian-Wolf theorem for multiple sources. 


\section{A Deterministic Model for Key Generation}

The idea in \cite{Wunder2016} was to utilize the local randomness $\omega_A$ and $\omega_B$ such that Alice and Bob send random signals over the channel. Therefore, instead of measuring the channel gain $K$ with pilot signalling alone, one gets a channel output $Y_A$ and $Y_B$ at Alice and Bob, respectively. Both of these are correlated via channel gain. To get some gain out of the local randomness, one also considers the local source of the sender. This means that Alice and Bob {\it virtually} receive $(Y_A,X_A)$ and $(Y_B,X_B)$, respectively. Now both sources are correlated in $K$, $X_A$ and $X_B$ and the following theorem was shown in \cite{Wunder2016}.
\begin{mytheorem}\label{Theorem1}
For general observations $Y_A=KX_B+Z_2$ and $Y_B=KX_A+Z_1$, of the products $KX_B$ and $KX_A$, with input $X_A,X_B \sim \mathcal{N}(0,P)$, channel gain $K \sim \mathcal{N}(0,\sigma_K^2)$ and noise $Z_1,Z_2 \sim \mathcal{N}(0,\sigma_Z^2)$ it holds that
\begin{IEEEeqnarray*}{rCl}
\IEEEeqnarraymulticol{3}{l}{
I(Y_A,X_A;Y_B,X_B)}\\
 &\geq & E_K[\log(1+\tfrac{|k|^2P}{\sigma^2_{Z}})]\\
&&-\: \tfrac{1}{2} E_{X_A}[\log (1+\tfrac{|x_A|^2\sigma_K^2}{\sigma^2_{Z}})]- \tfrac{1}{2} E_{X_B}[\log (1+\tfrac{|x_B|^2\sigma_K^2}{\sigma^2_{Z}})]\\
&&+\: \tfrac{1}{2} E_{X_A,X_B}\left[\log \left(1+ \frac{x_A^2x_B^2\sigma_K^4}{(x_A^2+x_B^2)\sigma_K^2\sigma_{Z}^2+\sigma_{Z}^4}\right)\right].
\end{IEEEeqnarray*}
\end{mytheorem}
The theorem shows that the new scheme splits its key generation gain between the sources of local randomness and the channel gain. The channel gain contribution is in the last line of theorem~1 and corresponds to state-of-the-art results. This shows that if the key generation rate from the channel gain is low, due to a long coherence time and therefore nearly constant channel gain, it is still possible to extract a non-zero key rate through the local randomness sources.

The proof for the achievability depends on classical typicality arguments and therefore lacks of simple ways to implement. The main challenge was a practical way to reconcile observations of both Alice and Bob. A simple solution to this problem, also presented in \cite{Wunder2016}, was to multiply the observations with the respective mutual local sources to produce correlated observations of a source $(X_AY_A,X_BY_B)$. This would yield a secure key rate of $I(Y_AX_A;Y_BX_A)$. However, 
exact calculation of the mutual information term $I(Y_AX_A;Y_BX_A)$ is involved, even for Gaussian signals. This is due to the multiplication operations which yield Bessel functions and to the best of our knowledge there is no known closed form solution for this term. However, we will approximate the term with our novel deterministic model to gain insights into its nature.

\subsection{The Linear Deterministic Model}
Lets look at a Gaussian point-to-point channel 
\begin{equation}
y=xh+z
\end{equation} 
where $z\sim \mathcal{N}(0,1)$ is Gaussian white noise and the input signal $x$ has an average input power constraint $E[|x|^2]\leq 1$. Both, the noise power and the input power are normalized to 1. Therefore, the instantaneous channel gain $h$, is scaled to $|h|=\sqrt{\text{SNR}}$. For the deterministic model, it is now assumed that the input signal $x$ and the noise $z$ have a peak power of one. Now one can represent  
both by a binary expansion which yields the following system model
\begin{equation}
y=2^{\tfrac{1}{2}\log \text{SNR}} \sum\limits_{i=1}^{\infty} x(i)2^{-i}+\sum\limits_{i=1}^{\infty} z(i)2^{-i}.
\end{equation}
One can group the bits under the decimal point together which yields 
\begin{equation}
y\approx 2^N \sum\limits_{i=1}^{N} x(i)2^{-i}+\sum\limits_{i=1}^{\infty} (x(i+n)+z(i))2^{-i}\label{approximation},
\end{equation}where $N=\lceil \tfrac{1}{2}\log \text{SNR} \rceil^+ $\footnote{For a complex Gaussian model, one can use the correspondence $N=\lceil \log \text{SNR} \rceil^+$}. Which is an approximation due to ignoring the 1-bit carry-over and due to restricting the channel gain to powers of two. 
The second sum can be cut-of since those bits are compromised by noise and it therefore yields a deterministic approximation. This means that the model takes an input $x=0.b_1b_2b_3\cdots$ and shifts the bits for $N$ positions over the decimal point $2^Nx=b_Nb_{N-1}\cdots b_1.b_0b_{-1}\cdots $. The noise just impairs the bits on the right of the decimal point $2^Nx+z=b_Nb_{N-1}\cdots b_1.\tilde{b}_0\tilde{b}_{-1}\cdots $, denoted as $\tilde{b}$. A deterministic approximation is achieved by cutting of the bit chain at the decimal point $y\approx b_N b_{N-1} \cdots b_1$. The resulting model can be written in an algebraic fashion as
\begin{equation}
\mathbf{y}=\mathbf{S}^{q-N}\mathbf{x},
\end{equation}
where $\mathbf{x}\in \mathbb{F}_2^q$ is a bit vector and $\mathbf{S}^{q-N}$ is a $q\times q-$shift matrix \begin{equation}
\mathbf{S}=\begin{pmatrix}
0 & 0 &  \cdots & 0 & 0\\
1 & 0 &  \cdots & 0 & 0\\
0 & 1 &  \cdots & 0 & 0\\
\vdots & \vdots & \ddots & \vdots & \vdots \\
0 & 0 &  \cdots & 1 & 0\\
\end{pmatrix}.
\end{equation}
With $\mathbf{S}$ an incoming bit vector can be shifted for $q-N$ positions. 
Note that we have used the approximation in \eqref{approximation} that $h=2^N\hat{h}\approx 2^N$ for $\hat{h}\in [1,2)$. However, we want to use the decimal bits of the channel gain as a source of randomness. Lets call $2^N$ the coarse channel gain and $\hat{h}$ the fine channel gain. The lower triangular model from \cite{Niesen-Ali} incorporates the fine channel gain into the model and represents it as a binary expansion, too. This yields a cauchy-product or discrete convolution between the bits of $x$ and $\hat{h}$ in \eqref{approximation}. It was proposed to use lower triangular toeplitz matrices to represent the fine channel gain, therefore the normal matrix-vector multiplication can represent the discrete convolution of the channel. As before, the resulting binary vector $\mathbf{y}$ gets cut off at noise level. We therefore have the following deterministic model
\begin{equation}
\mathbf{y}=\mathbf{H}\mathbf{x}
\end{equation}
with all operations over $\mathbb{F}_2$, $\mathbf{x}$ represents a bit vector and $\mathbf{H}$ is a square lower triangular Toeplitz matrix \begin{equation}
\mathbf{H}=\begin{pmatrix}
1 & 0 &  \cdots & 0 & 0\\
h_1 & 1 &  \cdots & 0 & 0\\
 \vdots & \vdots & \ddots & \vdots & \vdots \\
h_{N-2} & h_{N-3} & \cdots & 1 & 0\\   
h_{N-1} & h_{N-2} &  \cdots & h_1 & 1\\
\end{pmatrix}.
\end{equation}
where the $h_i$ represent the $N-1$ bits behind the decimal point of $\hat{h}=1.h_1h_2\cdots h_{N-1}$. We introduce the function $T_{lt}(\mathbf{x})=\mathbf{X}$ which maps a vector $\mathbf{x}$ to its square lower triangular Toeplitz matrix $\mathbf{X}$ and we introduce $T_{lt}^{-1}(\mathbf{X})=\mathbf{X}e_1=\mathbf{x}$ which maps a lower triangular Toeplitz matrix back to its vector. Here we used $e_1$ to indicate the first column vector of the identity matrix with dimension $\dim(\mathbf{X})$.

\subsection{A Deterministic Model for Key Exchange}
Let Alice send a real number $x_A(\omega_A)$ over the channel, while Bob sends $x_B(\omega_B)$ over the channel, where $\omega_A$ and $\omega_B$ represent local sources of randomness. We assume that both input signals have a unitary normalized power constraint $E[|x_A|^2],E[|x_B|^2]\leq 1$. We denote the channel gain between Alice and Bob by $k$. While keeping the coarse channel gain $2^N$ fixed, we let the fine channel gain $\hat{k}$ be uniformly distributed over the interval $[1,2)$ for all channel uses $n$. Note that this ensures that our generated key has maximal entropy and satisfies \eqref{Unif}. It is important to emphasize that we assume reciprocity between Alice and Bob {\it only} in the fine channel gain $\hat{k}$ and not in the coarse channel gains. The reason to allow differences in the coarse gain is to capture input power differences at the transceivers. We therefore denote the channel gain matrix from Bob to Alice by $\mathbf{K}'$.  We can model the system (fig. \ref{System model2}) in a deterministic way as 
\begin{IEEEeqnarray*}{rCl}
\mathbf{y}_B &=& \mathbf{K}\mathbf{x}_A\\
\mathbf{y}_A &=& \mathbf{K'}\mathbf{x}_B\\
\mathbf{y}_E &=& (\mathbf{H}_{1}\mathbf{x}'_A,\mathbf{H}_{2}\mathbf{x}'_B)
\end{IEEEeqnarray*}
where $\mathbf{y}_B$, $\mathbf{x}_A$ $\in \mathbb{F}_2^{N_A}$, $\mathbf{K} \in \mathbb{F}_2^{N_A\times N_A}$, $\mathbf{y}_A$, $\mathbf{x}_B$ $\in \mathbb{F}_2^{N_B}$, $\mathbf{K'} \in \mathbb{F}_2^{N_B\times N_B}$. Note that we assume a time division in which a transceiver can either receive or transmit. Both Alice and Bob alternate in receive and transmit mode. Alice uses the odd $n$ channel uses for transmission, while Bob uses the even $n$ channel uses for transmission. Therefore Eve observes
\begin{equation}
\mathbf{y}_E=\begin{cases}
\mathbf{H}_{1}\mathbf{x}'_A & \text{for } n \text{ odd} \\
\mathbf{H}_{2}\mathbf{x}'_B & \text{for } n \text{ even} 
\end{cases}
\end{equation} 
and so we have that $\mathbf{y}_E \in \mathbb{F}_2^{N_{1}}$ and $\mathbf{y}_E \in \mathbb{F}_2^{N_{2}}$ for $n$ odd and even, respectively. Moreover, $\mathbf{H}_{1} \in \mathbb{F}_2^{N_{1} \times N_{1}}$, $\mathbf{x}'_A \in \mathbb{F}_2^{N_{1}}$ and $\mathbf{H}_{2}\in \mathbb{F}_2^{N_{2} \times N_{2}}$, $\mathbf{x}'_B \in \mathbb{F}_2^{N_{2}}$. Note that we assume $N_1,N_2\leq\min\{N_A,N_B\}$ which is the same as assuming a channel advantage in case of classical wiretap models. We therefore guarantee a positive secure key rate in certain cases. Moreover, we have introduced a cut-off version of $\mathbf{x}_A $ and $\mathbf{x}_B$, denoted by $\mathbf{x}'_A $ and $\mathbf{x}'_B$, limited to $N_1$ and $N_2$ bits, respectively.

\subsection{Deterministic Security Constraints}

Due to the deterministic nature of the model for key exchange we can make simplifications on the security constraints. We define an achievable deterministic key rate $R_{d}$ if for every $\epsilon>0$ and sufficiently large $n$ there exists a strategy such that both generated keys at Alice and Bob, denoted by $\mathbf{s}_A$ and $\mathbf{s}_B$, respectively, satisfy
\begin{equation}
\Pr\{\mathbf{s}_A\neq \mathbf{s}_B\}=0,
\end{equation}
\begin{equation}
\tfrac{1}{n} H(\mathbf{s}_A) > R_{d}-\epsilon,
\end{equation}
\begin{equation}
\tfrac{1}{n} \log |\mathbf{s}_A| <\tfrac{1}{n} H(\mathbf{s}_A)+\epsilon,\label{Unif_det}
\end{equation}
where $|\mathbf{s}_A|$ counts the bits in the binary random vector $\mathbf{s}_A$. Moreover, we define a secure deterministic key rate $R_{sd}$ as the key rate $R_d$ with the following security constaint\begin{equation}\label{secure_det}
\tfrac{1}{n} I(\mathbf{y}_E^n;\mathbf{s}_A)=0.
\end{equation}
Note that we put a stricter notion on the difference between both keys. Moreover, we do not need a public communication channel. At last, Eq.~\eqref{secure_det} evokes a so-called perfect secrecy condition in contrast to weaker standard notions. All of these changes can be achieved with no further struggles due to the lack of noise in the model, which explains the modified security condition.

\subsection{Key Exchange by Pilot Signalling}
The state-of-the-art method for key exchange is to send a pilot signal over the channel. With the help of the pilot signal, one can measure the channel gain between Bob and Alice to extract a common key from these measurements. We can represent such a pilot signal by sending a basis vector $e_1$ over the channel with power $N$. We therefore have that $\mathbf{y}_B= \mathbf{K}\mathbf{x}_A= \mathbf{K}e_1= T^{-1}_{lt}(\mathbf{K})$ and $\mathbf{y}_B$ therefore contains the bits of $\mathbf{K}$. Sending a pilot signal to Alice results in her observation $\mathbf{y}_A$ containing the bits of $\mathbf{K'}$ as well. Note that $\mathbf{y}_A$ and $\mathbf{y}_B$ just differ in the number of bits $N_A$ and $N_B$\footnote{Note that our model assumes perfect CSI. Therefore, Alice and Bob know which bits are the same.}. We can use the observations as a key and get a deterministic key generation rate 
\begin{equation}
R_d=\frac{1}{2}\min \{N_A,N_B\},
\end{equation}
where the factor of $\tfrac{1}{2}$ is due to time-division, where both Alice and Bob need to transmit once to achieve the key rate. 
Linking the model back to the Gaussian model, using the correspondence $N=\lceil \tfrac{1}{2}\log \text{SNR} \rceil^+$, we get that the key rate~\footnote{This result resembles the one from \cite{LaiPoorKeyGen} for the high power regime in which $P\rightarrow \infty$ results in the key rate $R_{key}\sim \tfrac{1}{2T}\log P$ for $T=2$.}  is 
\begin{equation}
R=\tfrac{1}{4}\min \{\log \text{SNR}_A, \log \text{SNR}_B\},
\end{equation}
where $\text{SNR}_A$ represents the signal-to-noise ratio from Alice to Bob, and $\text{SNR}_B$ from Bob to Alice. Note that due to the independence of the channel gain between Alice and Bob, and both gains from Alice and Bob to Eve, the key generation rate is also a secure key generation rate $R_d=R_{sd}$. We see that using the deterministic model has two advantages. The first one is that due to the deterministic nature, no public communication is needed to reconcile the keys. The second advantage is that the binary expansion introduces a natural quantization which is fine enough to combat the noise. As a result, the observations at Alice and Bob can be used as a key without further post-processing. 
%

\subsection{Key Exchange by Product Signalling}

Assume that we do not send pilot signals over the channel, but generate a random number which is send over the channel. In that case both Alice and Bob receive a $\mathbf{y}$ which is the discrete convolution between the bits of a signal $\mathbf{x}$ and the bits of the channel gain $T^{-1}_{lt}(\mathbf{K})$. Generating a signal $\mathbf{x}_A$ and $\mathbf{x}_B$ and sending it over the respective channel produces two different observations $\mathbf{y}_A$ and $\mathbf{y_B}$. However, since the receivers Alice and Bob, know their own signal, they can multiply the observation from the left with 
$T_{lt}(\mathbf{x}_A)$ and  $T_{lt}(\mathbf{x}_B)$, respectively. 
This yields the following two observations
\begin{IEEEeqnarray*}{rCl}
T_{lt}(\mathbf{x}_B)\mathbf{y}_B &=& T_{lt}(\mathbf{x}_B)\mathbf{K}\mathbf{x}_A\\
T_{lt}(\mathbf{x}_A)\mathbf{y}_A &=&T_{lt}(\mathbf{x}_A) \mathbf{K'}\mathbf{x}_B.
\end{IEEEeqnarray*}

The following lemma will show, that both modified observations are the same.
\begin{lemma}
For arbitrary binary vectors $\mathbf{x}_A,\mathbf{y}_A \in \mathbb{F}_2^n$ and $\mathbf{x}_B,\mathbf{y}_B\in \mathbb{F}_2^m$, and truncated vectors $\bar{\mathbf{x}}_A,\bar{\mathbf{y}}_A,\bar{\mathbf{x}}_B,\bar{\mathbf{y}}_B\in \mathbb{F}_2^{\min\{n,m\}}$ we have that $T_{lt}(\bar{\mathbf{x}}_B)\bar{\mathbf{y}}_B=T_{lt}(\bar{\mathbf{x}}_A)\bar{\mathbf{y}}_A$.
\label{commutativity}
\end{lemma}\begin{proof}
First of all we note that the product of a lower triangular Toeplitz matrix with a vector is commutative. The operation mimics the product of two polynomials, where  the result follows from the commutativity of the product of polynomials. Alternatively, one can think about this product as a discrete convolution, which is also commutative. Truncating the matrix such that it has the same dimension as the vector squared does not change this fact. This means that for arbitrary binary vectors $\mathbf{x},\mathbf{y}\in \mathbb{F}_2^q$ we have that 
\begin{equation}\label{commuteMV}
T_{lt}(\mathbf{x})\mathbf{y}=T_{lt}(\mathbf{y})\mathbf{x}.
\end{equation}
We know that $T_{lt}(\mathbf{x})\mathbf{y}=T_{lt}(\mathbf{x})T_{lt}(\mathbf{y})e_1$ and that $T_{lt}(\mathbf{y})\mathbf{x}=T_{lt}(\mathbf{y})T_{lt}(\mathbf{x})e_1$ and we therefore have that

\begin{equation}
T_{lt}(\mathbf{x})T_{lt}(\mathbf{y})e_1=T_{lt}(\mathbf{y})T_{lt}(\mathbf{x})e_1,
\end{equation}
which shows that 
\begin{equation}\label{commuteMM}
T_{lt}(\mathbf{x})T_{lt}(\mathbf{y})=T_{lt}(\mathbf{y})T_{lt}(\mathbf{x})
\end{equation} since the first column of a lower triangular Toeplitz matrix determines the whole matrix. This shows the commutativity of the squared lower triangular matrices. Now we can proceed to show the lemma. 
We have that
\begin{IEEEeqnarray*}{rCl}
T_{lt}(\bar{\mathbf{x}}_B)\bar{\mathbf{y}}_B &=& T_{lt}(\bar{\mathbf{x}}_B)\bar{\mathbf{K}}\bar{\mathbf{x}}_A\\
&\overset{(a)}{=}& T_{lt}(\bar{\mathbf{x}}_B)T_{lt}(\bar{\mathbf{x}}_A)T_{lt}^{-1}(\bar{\mathbf{K}})\\
&\overset{(b)}{=}& T_{lt}(\bar{\mathbf{x}}_A)T_{lt}(\bar{\mathbf{x}}_B)T_{lt}^{-1}(\bar{\mathbf{K}})\\
&\overset{(c)}{=}& T_{lt}(\bar{\mathbf{x}}_A)T_{lt}(\bar{\mathbf{x}}_B)T_{lt}^{-1}(\bar{\mathbf{K'}})\\
&=& T_{lt}(\bar{\mathbf{x}}_A)\bar{\mathbf{K'}}\bar{\mathbf{x}}_B\\
&=& T_{lt}(\bar{\mathbf{x}}_A)\bar{\mathbf{y}}_A,
\end{IEEEeqnarray*}
where $(a)$ is due to eq.~\eqref{commuteMV}, $(b)$ is due to eq.~$\eqref{commuteMM}$ and $(c)$ is due to reciprocity of the fine channel gain. Note that the coarse channel gains match as well, since we are looking at a truncated channel gain matrix.
\end{proof}

This shows that we get the same observation at both receivers.
We can therefore use the observation as a common key. The key rate is $R_d=\tfrac{1}{2}\min\{N_A,N_B\}$. However, the secure key rate will be lower because Eve can receive a part of the locally generated randomness through her channel observations. We will consider some special cases of static channel gain to simplify the model and gain more insights.

\subsubsection{Static Channel Gain at all Transceivers}
We assume that all parties, Alice, Bob and Eve know the channel gains. Moreover, we assume that the channel gain is given by an identity matrix with dimension $N$, where the dimension represents the coarse channel gain. This means that the locally generated bit-vectors are received with a cut-off respective to the coarse channel gain. The intended case to model is, where all channels have a random but fixed (over $n$) channel gain, which is generated at the beginning of the $n$ channel uses. Note that modelling the channel gain by the identity matrix does not change the entropy of the observations in comparison to a fixed (constant) one, since fixed lower triangular toeplitz matrices are bijective mappings. The model is the following
\begin{IEEEeqnarray*}{rCl}
\mathbf{y}_B &=& \mathbf{I}_{N_A}\mathbf{x}_A\\
\mathbf{y}_A &=& \mathbf{I}_{N_B}\mathbf{x}_B\\
\mathbf{y}_E &=& (\mathbf{I}_{N_1}\mathbf{x}'_A,\mathbf{I}_{N_2}\mathbf{x}'_B),
\end{IEEEeqnarray*}
where $\mathbf{y}_B$, $\mathbf{x}_A$ $\in \mathbb{F}_2^{N_A}$, $\mathbf{I}_{N_A} \in \mathbb{F}_2^{N_A\times N_A}$, $\mathbf{y}_A$, $\mathbf{x}_B$ $\in \mathbb{F}_2^{N_B}$, $\mathbf{I}_{N_B} \in \mathbb{F}_2^{N_B\times N_B}$ and $\mathbf{y}_E \in \mathbb{F}_2^{N_{1}}$ and $\mathbf{y}_E \in \mathbb{F}_2^{N_{2}}$ for $n$ odd and even, respectively. Moreover, $\mathbf{I}_{N_1} \in \mathbb{F}_2^{N_{1} \times N_{1}}$, $\mathbf{x}'_A \in \mathbb{F}_2^{N_{1}}$ and $\mathbf{I}_{N_2}\in \mathbb{F}_2^{N_{2} \times N_{2}}$, $\mathbf{x}'_B \in \mathbb{F}_2^{N_{2}}$. 
We use product signalling, and Alice and Bob generate the keys $T_{lt}(\mathbf{x}_B)\mathbf{y}_B$ and $T_{lt}(\mathbf{x}_A)\mathbf{y}_A$, respectively. Both keys are the same due to Lemma \eqref{commutativity}. It can be easily seen that the deterministic
key generation rate is again 
\begin{equation}
R_d=\tfrac{1}{2}\min \{N_A,N_B\}.
\end{equation}
However, Eve can also observe both signals. The secure key rate is therefore dependent on the channel gain to Eve and resembles a wiretap scenario. Both signal sources are needed to construct the key, and the difference is inherently included in the bit-levels. It is therefore easy to see that the secure key rate is
\begin{equation}
R_{sd}= \tfrac{1}{2} (\min \{N_A,N_B \}-\min\{ N_1,N_2 \}).
\end{equation}
Linking the key rate to the Gaussian model gives 
\begin{IEEEeqnarray*}{rCl}
R_s &=& \tfrac{1}{4} \min \{\log \text{SNR}_A, \log \text{SNR}_B\}\\
&&- \min \{\log \text{SNR}_{E1}, \log \text{SNR}_{E2}\},
\end{IEEEeqnarray*}
where $\text{SNR}_{E1}$ and $\text{SNR}_{E2}$ denotes the channel gain to Eve at odd and even time slots, respectively.

\subsubsection{Static Channel Gain at Eve}
A natural extension to the previous case is to look into a model where we have a random varying channel gain in the legitimate channel and a constant gain for Eve. This is the worst case scenario from a physical layer security perspective, since Eve can receive all communication in plain, while Alice and Bob need to handle the channel gain as well. The model is the following
\begin{IEEEeqnarray*}{rCl}
\mathbf{y}_B &=& \mathbf{K}\mathbf{x}_A\\
\mathbf{y}_A &=& \mathbf{K'}\mathbf{x}_B\\
\mathbf{y}_E &=& (\mathbf{I}_{N_1}\mathbf{x}'_A,\mathbf{I}_{N_2}\mathbf{x}'_B),
\end{IEEEeqnarray*}
where $\mathbf{y}_B$, $\mathbf{x}_A$ $\in \mathbb{F}_2^{N_A}$, $\mathbf{K} \in \mathbb{F}_2^{N_A\times N_A}$, $\mathbf{y}_A$, $\mathbf{x}_B$ $\in \mathbb{F}_2^{N_B}$, $\mathbf{K'} \in \mathbb{F}_2^{N_B\times N_B}$. Moreover, $\mathbf{I}_{N_1} \in \mathbb{F}_2^{N_{1} \times N_{1}}$, $\mathbf{x}'_A \in \mathbb{F}_2^{N_{1}}$ and $\mathbf{I}_{N_2}\in \mathbb{F}_2^{N_{2} \times N_{2}}$, $\mathbf{x}'_B \in \mathbb{F}_2^{N_{2}}$.  It is easy to see that we cannot achieve a higher rate $R_d$ than the previous cases with pilot signalling or static channel gain, since the maximum number of bits in the key vector is upper bounded by the mutual coarse channel gain $\min\{N_A,N_B\}$. However, calculating the secure key generation rate is more involved, and one needs to look into the term $I(\mathbf{y}_E^n;\mathbf{s}_A)$. We investigate the mutual information for a specific time step: $I(\mathbf{x}'_A,\mathbf{x}'_B;T_{lt}(\mathbf{x}_B)\mathbf{y}_B)$. If we assume that $N_A=N_B=N_1=N_2$, we can show the following: 
\begin{IEEEeqnarray*}{rCl}
I(\mathbf{x}'_A,\mathbf{x}'_B;\mathbf{X}_B\mathbf{y}_B)
&=&h(\mathbf{X}_B\mathbf{y}_B)-h(\mathbf{X}_B\mathbf{y}_B|\mathbf{x}'_B,\mathbf{x}'_A)\\
&\overset{(a)}{=}&h(\mathbf{X}_B\mathbf{y}_B)-h(\mathbf{X}_B\mathbf{X}_A\mathbf{k}|\mathbf{x}_B,\mathbf{x}_A)\\
&=&h(\mathbf{X}_B\mathbf{X}_A \mathbf{k})-h(\mathbf{k}),
\end{IEEEeqnarray*}
where we denote $T_{lt}(\mathbf{x}_B)=\mathbf{X}_B$, $T_{lt}(\mathbf{x}_A)=\mathbf{X}_A $ and $T_{lt}^{-1}(\mathbf{K})=\mathbf{k}$. Note that (a) is due to \eqref{commuteMV} and the assumptions on $N_1,N_2,N_A,N_B$.
The multiplication by $\mathbf{X}_B\mathbf{X}_A$ is an bijection in the case of known $\mathbf{x}_B$ and $\mathbf{x}_A$. In that case we would have that $I(\mathbf{x}'_A,\mathbf{x}'_B;T_{lt}(\mathbf{x}_B)\mathbf{y}_B)=0$, fulfilling the secrecy constraint. This suggest a secrecy protocol which only uses local bit-levels as additional source of randomness, if those bit-levels are not received at Eve. For this purpose we can divide the local randomness vectors $\mathbf{x}_A$ and $\mathbf{x}_B$ in common and private parts, where the common part can be received by the legitimate receiver, as well as by Eve. The private part on the other hand is only received by the legitimate receiver. Both signals can then be partitioned into two parts $\mathbf{x}_A=\mathbf{x}_A^p+\mathbf{x}_A^c$ and $\mathbf{x}_B=\mathbf{x}_B^p+\mathbf{x}_B^c$, where $\mathbf{x}_A^p,\mathbf{x}_B^p>0$ if and only if $N_A>N_1$ and $N_B>N_2$. We can design the send signal such that we only use the private part of the signal to send random bits, and the common part for pilot signalling, for example \vspace{-1.5em}\begin{equation*}
\mathbf{x}_A=\underbracket{100\cdots0\overbracket{b_1b_2\cdots b_{N_A-N_1}}^{\mathbf{x}_A^p}}_{N_A \text{ bits}}.
\end{equation*}
Due to the lower triangular structure of the channel gain operation, the private bits only get down-shifted in the observations. One can then split $I(\mathbf{x}'_A,\mathbf{x}'_B;\mathbf{X}_B\mathbf{y}_B)$ into a common $I(\mathbf{x}'_A,\mathbf{x}'_B;\mathbf{k}^c)$ and a private part $I(\mathbf{x}'_A,\mathbf{x}'_B;\mathbf{X}_B\mathbf{y}_B|\mathbf{k}^c)=I(\mathbf{x}'_A,\mathbf{x}'_B;(\mathbf{X}_B\mathbf{y}_B)^p)$ and show that both are zero. In this way we have exploited the structure of the deterministic model to design a scheme which uses a form of mixed signalling, where the common parts utilize pilot signalling and the private parts utilize product signalling. Due to the assumption that Eve has a static channel gain and can therefore see the local contribution in plain, we have obtained a worst-case scenario with a minimal achievable secure key rate. 
\subsection{Discussion}
We have analysed both, the state-of-the-art pilot signalling scheme and the new product signalling with a deterministic model. We have shown that the general key generation rate is the same for both schemes for a perfect channel gain behaviour, i.e. uniformly distributed with short coherence time. This is due to the fact that the overall size of the bit-vectors stays the same. Therefore, product signalling would have no advantage compared to pilot signalling. Moreover, the secure key rate for product signalling can be even worse because Eve can listen to both Alice and Bob, and therefore gets parts of the local randomness sources. This means that there is a trade-off which closely resembles that of a wiretap scenario and we have proposed a scheme to exploit the created algebraic structure. Product signalling begins to shine in cases with long coherence time. Here, one can compensate the lack of randomness in the channel gain, by feeding in the local sources. Product signalling would therefore yield a more robust key generation technique.

\section{Conclusions}
Motivated by an open problem in \cite{Wunder2016}, we have developed a deterministic model for secure key rate analysis of Gaussian models. The approximation is used to show secure key generation rate results on a product signalling scheme, developed in \cite{Wunder2016}. The proposed approximate model provides insights which were out-of-reach within the classical Gaussian model. An advantage of the new model is that, due to its deterministic nature (i.e. absence of noise), the key rate can be achieved without a public communication channel. Moreover, the model has an inherent quantization, which makes it possible to directly derive key rates from the equations. An interesting part is the additional algebraic structure. It was shown in the past, that algebraic structures can be exploited in several ways to gain unexpected results, especially for multi-user networks. Future research could therefore look into application of our model to analyse multi-user key generation scenarios. Furthermore, there is a need to investigate the exact gap between the approximate rate and the corresponding Gaussian model. We expect that this gap is within a few bits, due to similar results in several works on the linear deterministic model, e.g. \cite{Bresler2008}. Moreover, rate leakage in the noise effected part of the signal could lead to an adjustment of the secure key rate  of the corresponding Gaussian model. Nonetheless, we believe that the proposed model can unlock some previously out-of-reach results and therefore act as a powerful tool for the analysis of secure key generation problems.

\bibliographystyle{./IEEEtran}
\bibliography{./ref}

\end{document}